\documentclass[11pt,letterpaper]{article}

\usepackage{amsmath, amssymb, amsthm}

\usepackage{bookmark}
\usepackage{tikz}
\usepackage{float} % Add this in the preamble

\usepackage{tikzit}
\usepackage{import}
\usepackage[T1]{fontenc}
\usepackage{xifthen}
\usepackage{pdfpages}
\usepackage{transparent}
\usepackage{geometry}
\usepackage{thm-restate}

\usepackage{titling}
\usepackage{lipsum} % For placeholder abstract
\usepackage{algorithm}
\usepackage{algpseudocode}
\usepackage{hyperref} % This package allows you to refer within the document with clickable hyperlinks.
 \hypersetup{colorlinks = true,
     linkcolor={black},
     citecolor={blue},
     urlcolor={magenta}} % This command allows you to change how links within the document appear
\usepackage{tikz} %This package is for creating graphics. I am including it
\usepackage{float}
\usepackage{ytableau}
\usepackage[capitalize]{cleveref}
\usepackage{cite}

\algtext*{EndIf}
\algtext*{EndWhile}
\algtext*{EndFunction}

\usepackage{mathtools}
\theoremstyle{plain} %theorems after this line but before the definition
\newtheorem{theorem}{Theorem}[section]
\newtheorem{lemma}[theorem]{Lemma}
\newtheorem{corollary}[theorem]{Corollary}

\theoremstyle{definition} %theorems after this line will all have the same
\newcommand{\TokenSwapping}{\textsc{TokenSwapping}}
\newcommand{\WeightedTokenSwapping}{\textsc{WeightedTokenSwapping}}
\newcommand{\HappySwapAlgo}{\textsc{HappySwapAlg}}
\newcommand{\ExtendedCycleAlgo}{\textsc{ExtendedCycleAlg}}

\title{Settling Weighted Token Swapping up to Algorithmic Barriers}
 \author{Nicole Wein\\University of Michigan\thanks{\texttt{nswein@umich.edu}} \and Guanyu (Tony) Zhang\\University of Michigan\thanks{\texttt{guanyuzh@umich.edu}}}
\date{}

\begin{document}

\maketitle

\begin{abstract}
We study the \emph{weighted token swapping} problem, in which we are given a graph on $n$ vertices, $n$ weighted tokens, an initial assignment of one token to each vertex, and a final assignment of one token to each vertex. The goal is to find a minimum-cost sequence of swaps of adjacent tokens to reach the final assignment from the initial assignment, where the cost is the sum over all swaps of the sum of the weights of the two swapped tokens.

\emph{Unweighted} token swapping has been extensively studied: it is NP-hard to approximate to a factor better than $14/13$, and there is a polynomial-time 4-approximation, along with a tight ``barrier'' result showing that the class of \emph{locally optimal} algorithms cannot achieve a ratio better than 4. For trees, the problem remains NP-hard to solve exactly, and there is a polynomial-time 2-approximation, along with a tight barrier result showing that the class of \emph{$\ell$-straying} algorithms cannot achieve a ratio better than 2.

\emph{Weighted} token swapping with $\{0,1\}$ weights is much harder to approximation: it is NP-hard to approximate even to a factor of $(1-\varepsilon) \cdot \ln n$ for any constant $\varepsilon>0$. Restricting to \emph{positive} weights, no approximation algorithms are known, and the only known lower bounds are those inherited directly from the unweighted version.

We provide the first approximation algorithms for weighted token swapping on both trees and general graphs, along with \emph{tight} barrier results. Letting $w$ and $W$ be the minimum and maximum token weights, our approximation ratio is $2+2W/w$ for general graphs and $1+W/w$ for trees.

\end{abstract}
\clearpage
% \pagenumbering{gobble}
% \pagebreak
% \newpage
% \pagenumbering{arabic}
\section{Introduction}

In the \emph{weighted token swapping} problem, we are given an $n$-vertex graph $G=(V,E)$, $n$ distinct tokens each with a non-negative weight, and two one-to-one assignments
of tokens to vertices: a starting assignment, and a destination assignment. A swap along an edge
$(u, v) \in E$ switches the locations of the tokens on vertices $u$ and $v$. The \emph{cost} of a swap is the sum of the weights of the two swapped tokens. The token swapping problem asks
for the minimum-cost sequence of swaps to arrive at the destination assignment from the starting assignment.

The \emph{unweighted} version of the problem (where the goal is to minimize the total number of swaps) has been extensively studied in both theory and practice~\cite{aicholzer:2021, MR4541302, akers:1989, MR2431751, MR1137822, MR1705338, MR1334632, YAMANAKA:2015, miltzow_et_al:LIPIcs.ESA.2016.66, MR3805577, cayley1849lxxvii,MR796304,MR1691876,MR3917574,yasui2015swapping,portier1990whitney,MR4705861,hiken2024improvedhardnessofapproximationtokenswapping,ajila2024colored,MR4036097,surynek2019multi,gourves2017object,MR1285588,MR3710080,MR1666061,banerjee2022locality, siraichi2019qubit, molavi2022qubit,bapat2023advantages,MR3964104,MR4638397,sharma2023noise,MR4594484,MR4261033,surynek2018finding}. It has found practical relevance in areas such as network engineering~\cite{akers:1989}, robot motion planning~\cite{MR4036097,surynek2019multi}, game theory~\cite{gourves2017object}, and --- in the case of the parallel variant --- qubit routing (e.g.~\cite{banerjee2022locality, siraichi2019qubit, molavi2022qubit,bapat2023advantages,MR3964104,MR4638397,sharma2023noise,MR4594484}). Algorithms and heuristics for the problem have also been experimentally evaluated~\cite{MR4261033,surynek2018finding}. In terms of theoretical results, unweighted token swapping is NP-hard to approximate to better than a $14/13$ factor~\cite{hiken2024improvedhardnessofapproximationtokenswapping,miltzow_et_al:LIPIcs.ESA.2016.66}, but admits a polynomial-time 4-approximation algorithm~\cite{miltzow_et_al:LIPIcs.ESA.2016.66}.  

See the introductions of~\cite{hiken2024improvedhardnessofapproximationtokenswapping,MR4541302,aicholzer:2021} for more background on unweighted token swapping, including parameterized algorithms and hardness, special classes of graphs, and the parallel variant. Our focus is on \emph{weighted} token swapping.

Weighted token swapping has been studied in several prior works~\cite{hiken2024improvedhardnessofapproximationtokenswapping,aicholzer:2021,MR4541302}. When the weights are only 0 and 1, weighted token swapping becomes much harder to approximate than unweighted: it is NP-hard to approximate even within a factor of $(1-\varepsilon) \cdot \ln n$ for any constant $\varepsilon>0$~\cite{hiken2024improvedhardnessofapproximationtokenswapping}. However, when the weights are restricted to be positive, nothing is known! That is, there are no known approximation algorithms, and the only known lower bounds are inherited directly from the unweighted case.

We obtain the first algorithms along with tight ``barrier'' results for weighted token swapping on both general graphs and trees. To provide context for our results, we will first survey the analogous results for \emph{unweighted} token swapping, starting with the case of trees.

\textbf{Unweighted Token Swapping on Trees.} Token swapping on trees was actually studied before token swapping on general graphs, starting with Akers and Krishnamurthy who introduced the problem in 1989 in the context of network engineering~\cite{akers:1989}. Interestingly, token swapping on trees was independently introduced three different times, and each time a polynomial-time 2-approximation was discovered. All three algorithms are different, and are known as the Happy Swap Algorithm~\cite{akers:1989}, the Cycle Algorithm~\cite{YAMANAKA:2015}, and the Vaughan-Portier Algorithm~\cite{MR1334632}. More recently, token swapping on trees was shown to be NP-complete~\cite{aicholzer:2021}. Regarding the approximation ratio of 2, a tight ``barrier'' result is known: roughly speaking, a better-than-2-approximation for trees is impossible, unless the algorithm exhibits ``strange'' behavior. Specifically, an \emph{$\ell$-straying} algorithm is defined as one in which for every input, every token stays within a distance $\ell$ from its (unique) path from starting vertex to destination vertex. Counterintuitively, any better-than-2-approximation algorithm must have a huge straying value: at least $\Omega(n^{1-\varepsilon})$ for every constant $\varepsilon>0$~\cite{aicholzer:2021}.

\textbf{Unweighted Token Swapping on General Graphs.}  Miltzow,
Narins, Okamoto, Rote, Thomas, and Uno~\cite{miltzow_et_al:LIPIcs.ESA.2016.66} showed that token swapping is APX-hard and admits a polynomial-time 4-approximation. Hiken and Wein showed that it is NP-hard to approximate to better than a $14/13$ factor~\cite{hiken2024improvedhardnessofapproximationtokenswapping}. As for barrier results, the notion of $\ell$-straying is not meaningful for general graphs because any valid algorithm must be $\Omega(n)$-straying (consider the example of a cycle where each token wants to shift over by 1). For this reason, prior work considers a different natural class of algorithms: a \emph{locally optimal} algorithm is one that never performs a swap that takes \emph{both} tokens farther from their destinations. Hiken and Wein showed the tight barrier result that no locally optimal algorithm can yield a better-than-4-approximation~\cite{hiken2024improvedhardnessofapproximationtokenswapping}. 

\subsection{Our Results}

We obtain the first approximation algorithms for weighted token swapping on both trees and general graphs, which generalize the best known results for unweighted graphs. To complement these algorithms, we show \emph{tight} barrier results analogous to those for unweighted graphs.

Our bounds are in terms of $w$ and $W$, the minimum and maximum token weights, respectively. As a baseline, any $\alpha$-approximation algorithm for \emph{unweighted} token swapping is automatically an $(\alpha\cdot W/w)$-approximation for \emph{weighted} token swapping. This is because every swap that the algorithm performs has cost at most $2W$, whereas any swap that the optimal algorithm performs has cost at least $2w$; so, if the algorithm does at most $\alpha$ times more swaps than the optimum, then the weighted cost is at most $\alpha\cdot W/w$ times the optimum. As a result, the best known unweighted algorithms imply a $2W/w$-approximation for weighted token swapping on trees, and a $4W/w$-approximation for weighted token swapping on general graphs. We show that one can do better for both trees and general graphs.

\subsubsection{Our Results for Trees}
Our first result for trees is a $(1+W/w)$-approximation algorithm:

\begin{restatable}{theorem} {treealg}\label{treealg}
There is a polynomial-time $(1+W/w)$-approximation algorithm for weighted token swapping on trees.
\end{restatable}

Note that by setting $W=w$ to recover the unweighted case, this approximation ratio is consistent with the best known approximation ratio of 2.\\

We complement our algorithm with essentially the strongest possible barrier result for $\ell$-straying algorithms (defined above):

\begin{restatable}{theorem} {treelb}\label{treelb}
For any constants $\varepsilon,\delta>0$, there does not exist an $O(n^{1-\varepsilon})$-straying algorithm that provides a $(1+W/w-\delta)$-approximation for weighted token swapping on trees.
\end{restatable}

In contrast, our algorithm that achieves a $(1+W/w)$-approximation is only 1-straying.

\subsubsection{Our Results for General Graphs}

We obtain a $(2+2W/w)$-approximation for general graphs:

\begin{restatable}{theorem} {genalg}\label{genalg}
There is a polynomial-time $(2+2W/w)$-approximation algorithm for weighted token swapping on general graphs.
\end{restatable}

Note that by setting $W=w$ to recover the unweighted case, this approximation ratio is consistent with the best known approximation ratio of 4.\\

Following prior work, it would be natural to show a barrier result for the class of locally optimal algorithms (defined above). However, there is a caveat: there is a known token swapping algorithm that is \emph{not} locally optimal. It is a 4-approximation algorithm for unweighted token swapping from  the appendix of~\cite{hiken2024improvedhardnessofapproximationtokenswapping} that was presented as an alternative to the known 4-approximation~\cite{miltzow_et_al:LIPIcs.ESA.2016.66}. To obtain a more robust barrier, we present a generalized definition of locally optimal that encompasses all known token swapping algorithms:

A token swapping  algorithm is \emph{generalized locally optimal} if every swap takes at least one of the two participating tokens closer to either (a) its destination vertex or (b) its starting vertex. (Locally optimal algorithms omit condition (b) and are thus a strict subset of generalized locally optimal algorithms.)

Generalized locally optimal algorithms are a large and natural class of algorithms, however, one may initially wonder why it would be natural to move a token closer to its starting vertex. The intuition is that it is useful to allow a token $t$ to sit on its starting vertex, get pushed off that vertex by a token in transit, and then return to its starting vertex, repeatedly, until it is $t$'s ``turn'' to move towards its destination. In this scenario, all of the swaps that move $t$ back onto its starting vertex are permitted by a generalized locally optimal algorithm.

We prove a barrier result for generalized locally optimal algorithms that is \emph{tight} with our algorithm from \cref{genalg}.

\begin{restatable}{theorem} {genlb}\label{genlb}
For any constant $\delta>0$, there does not exist a generalized locally optimal algorithm that provides a $(2+2W/w-\delta)$-approximation for weighted token swapping on general graphs.
\end{restatable}

Additionally, as shown in \cref{sec:opt}, our $(2+2W/w)$-approximation algorithm is generalized locally optimal (as are all previously known token swapping algorithms).

\subsection{Remarks about our Proofs}

\paragraph{Simplicity.} Three of our four proofs are self-contained and simple. The fourth (\cref{genlb}) is also quite simple but relies on a 3-page proof from prior work~\cite{hiken2024improvedhardnessofapproximationtokenswapping}. Additionally, the algorithms for \cref{treealg,genalg} are both simple greedy algorithms.

Furthermore, \cref{treelb} is a simplification of prior work. Specifically, the unweighted construction for the $O(n^{1-\varepsilon})$-straying barrier result from prior work is more involved than ours; ours is simply a path with two stars on each end (\cref{fig:straying-approximation}). This subsumes the result of prior work~\cite{aicholzer:2021} and extends it to the weighted setting, while simplifying it. 

\paragraph{Comparison of our Algorithms to Known Unweighted Algorithms.} Recall that for trees, there are three known 2-approximation algorithms: the Happy Swap Algorithm~\cite{akers:1989} (see also~\cite{miltzow_et_al:LIPIcs.ESA.2016.66}), the Cycle Algorithm~\cite{YAMANAKA:2015}, and the Vaughan-Portier Algorithm~\cite{MR1334632}. The Vaughan-Portier Algorithms is similar-in-spirit to the Happy Swap Algorithm, and we will restrict our discussion to the Happy Swap and Cycle Algorithms. 

Our algorithm for trees is, perhaps surprisingly, just the Happy Swap Algorithm with no modifications whatsoever. One simply runs the algorithm while ignoring the weights of the tokens. This, together with the barrier result, shows that to improve the approximation ratio below our bound of $1+W/w$, one cannot use any tricks like prioritizing swaps of tokens of certain weight, without also making a radical change to the structure of the algorithm that renders it $\Omega(n^{1-\varepsilon})$-straying.

Turning our attention to general graphs, there are known extensions of both the Happy Swap Algorithm~\cite{miltzow_et_al:LIPIcs.ESA.2016.66} and the Cycle Algorithm~\cite{hiken2024improvedhardnessofapproximationtokenswapping} that are both 4-approximations. Given our tree result, one may naturally expect that the extension of the Happy Swap Algorithm would yield a good algorithm for general graphs. This is not what we show, however. We briefly describe why this algorithm seems not to admit a straightforward analysis in the weighted setting: There could be situation where a token $t$ is on its destination vertex, then moved off its destination, and then later becomes part of a cycle of tokens that all want to shift over by one. To resolve this cycle, the extension of the Happy Swap Algorithm chooses one token to go all the way around the cycle. One natural way to bound the approximation ratio would be to establish a relationship between the weight of the token that moved $t$ off of its destination, and the weight of the token that goes around the cycle to put $t$ back onto its destination. However, these two actions could happen at completely different times in the algorithm, making them difficult to compare. In particular, the troublesome case is when the original token that moved $t$ is light, while the eventual cycle is composed of only heavy tokens.

Luckily, the extension of the Cycle Algorithm for general graphs does not suffer from this drawback. Instead, it is more ``temporally local'' in that when a token is moved from its destination, it is quickly returned. For this reason, the Cycle Algorithm for general graphs is simpler to analyze than the Happy Swap Algorithm for general graphs, for weighted token swapping. However, the Cycle Algorithm for general graphs does not work straight out of the box, and requires a minor tweak: while the original algorithm picks an arbitrary token to go around a given cycle, our variant picks the smallest-weight token on the cycle. 

Contextually, the Cycle Algorithm for general graphs was introduced recently in the appendix of~\cite{hiken2024improvedhardnessofapproximationtokenswapping} as an equally attractive alternative to the Happy Swap Algorithm for general graphs. However, we show that it has the unintended consequence of lending itself particularly nicely to the weighted setting.

\section{Preliminaries}
Throughout this paper, we adhere to the following notations and assumptions: A \TokenSwapping\ instance consists of parameters \TokenSwapping{$(G(V,E), T, v_0, v_f)$}. $G(V,E)$ is the underlying undirected simply-connected graph with vertex set $V$, edge set $E$, and we let $n = |V|$. We denote $T$ as the set of tokens such that there is exactly one token on each vertex during the execution of an algorithm for \TokenSwapping, henceforth $|T| = |V| = n$. We also have two bijections $v_0, v_f: T \xrightarrow{} V$, such that for any $t\in T$, $v_0(t)$ and $v_f(t)$ are the starting vertex and the destination vertex of token $t$, respectively. A \WeightedTokenSwapping\  instance refers to \WeightedTokenSwapping{$(G(V,E), T, v_0, v_f, \omega)$} which contains all the parameters in \TokenSwapping\ as well as a weighting function $\omega : T \rightarrow \mathbb{R}_{> 0}$. Denote $w = \min\limits_{t\in T}\omega(t)$, $W = \max\limits_{t\in T}\omega(t)$ and $v: T \xrightarrow{} V$ as the bijection where for each $t \in T$, $v(t)$ is the vertex that token $t$ currently occupies, and let $d(v_1, v_2)$ be the distance between vertices $v_1$ and $v_2$, let $d(t_1, t_2)$ be a shorthand for $d(v(t_1), v(t_2))$. 

In both \TokenSwapping\ and \WeightedTokenSwapping, a swap on incident tokens $(t_1, t_2)$ is said to be \textbf{valid} if $(v(t_1), v(t_2)) \in E$, and the \textbf{cost} for a valid swap on tokens $(t_1, t_2)$ in \WeightedTokenSwapping\ is $\omega(t_1) + \omega(t_2)$.  The goal for \TokenSwapping\ is to find a sequence of valid swaps, starting from the configuration $v(t) = v_0(t)$ for all $t \in T$ and ending in the configuration $v(t) = v_f(t)$ for all $t\in T$, such that the number of swaps in the sequence achieves minimum among all such swap sequences.\WeightedTokenSwapping\ is the same except the goal is to minimize the total cost of the sequence of swaps. We will denote OPT as the optimal cost of \TokenSwapping\ or \WeightedTokenSwapping\ , depending on the context, and ALG as the cost of the algorithm that is being discussed to solve \TokenSwapping\  or \WeightedTokenSwapping, depending on the context.

\section{Results for Trees}

 First, we study the approximation algorithms of \WeightedTokenSwapping\  on \textbf{trees}. 
 %A basic property of trees is: for any $v_1, v_2 \in V$, there exists a unique path between $v_1$ and $v_2$. 
 For any token $t \in T$, we denote $p(t)$ as the unique path from $v_0(t)$ to $v_f(t)$, 
 %viewed as an induced subgraph of $G(V,E)$, which is also the shortest path for token $t$ to reach $v_f(t)$ starting from $v_0(t)$, 
 and let $d(t) := d(v_0(t), v_f(t))$ be the distance between these two vertices. 
\subsection{A polynomial time ($1 + \frac{W}{w}$)-approx. algorithm}

In this section we will prove the following theorem:

\treealg*

% Previously, people have shown that the so-called \textbf{Happy Swap Algorithm} achieves a 2-approximation on tree instances of \TokenSwapping\ \cite{akers:1989}, see also \cite{miltzow_et_al:LIPIcs.ESA.2016.66}. In this section, we will show a generalized approximation result on tree instances of \WeightedTokenSwapping\  parametrized by the lower bound and upper bound of the weights of tokens. More specifically, we will see that \HappySwapAlgo\ achieves a (1+$\frac{W}{w}$)-approximation on tree instances of \WeightedTokenSwapping, henceforth, the previous result on \TokenSwapping\ is subsumed by our result if we set $W=w$. 

\subsubsection{Algorithm}
The algorithm that realizes \cref{treealg} is simply the known \textbf{Happy Swap Algorithm}~\cite{akers:1989}. For convenience, we specify \HappySwapAlgo\ below as well as its related definitions; for a proof of the correctness of \HappySwapAlgo, we refer the readers to \cite{miltzow_et_al:LIPIcs.ESA.2016.66}. 

Following the notation of \cite{MR4541302}, we say that a token $t \in T$ is \textbf{happy} if $v(t) = v_f(t)$, i.e.~$t$ is at its destination. A swap on tokens $(t_1, t_2)$ is said to be \textbf{happy} if $d(v_f(t_1), v(t_2)) = d(v_f(t_1), v(t_1)) - 1$ and $d(v_f(t_2), v(t_1)) = d(v_f(t_2), v(t_2)) - 1$, in other words, the swap brings both of the tokens closer to their destinations. A swap on tokens $(t_1, t_2)$ is said to be a \textbf{shove} if one of them is happy, say $v(t_1) = v_f(t_1)$, and the other token decrements the distance to its destination, i.e., $d(v_f(t_2), v(t_1)) = d(v_f(t_2), v(t_2)) - 1$.

To run the \HappySwapAlgo\  on an instance of \WeightedTokenSwapping, we simply run the algorithm as is, ignoring the weights of the tokens.

\begin{algorithm}[H]
    \caption{Happy Swap Algorithm}
\begin{algorithmic}[1]
\Function{\HappySwapAlgo}{\TokenSwapping($G(V,E), T, v_0, v_f$)}:
\While{there exists a token $t \in T$ such that $v(t) \neq v_f(t)$}
\If{there exists a happy swap \textbf{or} there exists a shove on a pair of tokens $(t_1, t_2)$}
    \State perform a swap on $(t_1, t_2)$
\EndIf
\EndWhile
\EndFunction
\end{algorithmic}
\end{algorithm}

% \begin{lemma}
% \label{happyswapproof}
%     At any time during the execution, there exists either a happy swap or a shove in \TokenSwapping, hence the cases in the above \textit{while} loop are exhaustive. Moreover, \HappySwapAlgo\  terminates in polynomial time and induces a 2-approximation on \TokenSwapping. Additionally, \HappySwapAlgo\ is a 1-straying algorithm.  
% \end{lemma}
% \begin{proof}
%     See \cite{miltzow_et_al:LIPIcs.ESA.2016.66}.
% \end{proof}

Prior work has shown that the \HappySwapAlgo\ runs in polynomial time and correctly solves \TokenSwapping, see \cite{akers:1989, miltzow_et_al:LIPIcs.ESA.2016.66}. Thus, it remains to show that it is a ($1 + \frac{W}{w}$)-approx.

\subsubsection{Approximation Analysis}

 A \textbf{move} is a token-vertex pair $(t, v)$ which indicates the operation of moving token $t$ into vertex $v$, and we will use the shorthand $t \xrightarrow{} v$ for notational convenience. Now we focus on an individual swap on a pair of tokens $(t_1, t_2)$ incurred by \HappySwapAlgo. Note that there are two moves $t_1 \rightarrow v(t_2)$ and $t_2 \rightarrow v(t_1)$ in the swap along the edge $(v(t_1), v(t_2)) \in E$. We say a move $t \xrightarrow{} v$ is \textbf{inevitable} if $t \xrightarrow{} v$ occurs in the first $d(t)$ moves on token $t$, otherwise the move $t \xrightarrow{} v$ is \textbf{redundant}.

\begin{lemma}
\label{cannot_stray}
     If a move $t \xrightarrow{} v$ is inevitable, then $(v(t),v) \in p(t)$. In other words, inevitable moves cannot cause the token $t$ to stray away from the path $p(t)$. 
\end{lemma}
\begin{proof}
    Suppose $p(t)$ has the form $(v_0(t) = v_0, v_1, \cdots, v_{d(t)} = v_f(t))$, and for the sake of contradiction, assume some inevitable swap $s$ on token $t$ moves it to some other vertex $v' \not \in p(t)$ from some $v_i$ where $0 \leq i < d(t)$ (note that $i \neq d(t)$ by the definition of being inevitable, since $t$ must have been moved at least $d(t)$ times to reach $v_{d(t)}$).

Since swap $s$ moves $t$ off of $p(t)$, swap $s$ does not move $t$ closer to its destination $v_f(t)$. Since happy swaps  move both tokens closer to their destinations, and shoves move one token closer to its destination, this means that swap $s$ must have been a shove, with $t$ as the initially happy token. However, $t$ is not happy since $0 \leq i < d(t)$.
\end{proof}

\begin{corollary}
\label{distance_one}
    If a move of token $t$ is redundant, then it can only have the form $t \rightarrow v'$ where $d(v', v_f(t)) \leq 1$.
\end{corollary}
\begin{proof}
    Since the first $d(t)$ moves (inevitable moves) on token $t$ cannot make $t$ leave the path $p(t)$ by Lemma \ref{cannot_stray}, and observe that \HappySwapAlgo\  is 1-straying since the only way a happy swap or shove can move a token farther from its destination is if $t$ was happy right before the swap, we know that a redundant move on token $t$ can only occur on some edge that is connected with $v_f(t)$.
\end{proof}

\begin{lemma}
\label{upper_bound}
    At least one of the two moves of any swap incurred by \HappySwapAlgo\  must be inevitable.
\end{lemma}
\begin{proof}
Suppose $t_1$ is on vertex $v_1$, $t_2$ is on vertex $v_2$, and $t_1, t_2$ are about to swap with each other.\\

    \textbf{Case I:} The swap on $(t_1, t_2)$ is a happy swap.\\
    For the sake of contradiction, suppose both of the moves $t_1 \rightarrow v_2$ and $t_2 \rightarrow v_1$ are redundant. By the definition of happy swap, we know $v_1 \neq v_f(t_1)$ and $v_2 \neq v_f(t_2)$, i.e.~neither token is at its destination right before the swap. Henceforth, according to Corollary \ref{distance_one}, we know $v_2 = v_f(t_1)$ and $v_1 = v_f(t_2)$, and both of the tokens have reached their destination vertices at least once since we assumed the moves are both redundant, in other words, $t_1$ was shoved away from $v_2$ onto $v_1$, and $t_2$ was shoved away from $v_1$ onto $v_2$. Observe that it cannot be the case that $t_1$ arrived at $v_1$ and $t_2$ arrived at $v_2$ simultaneously, since otherwise this implies a swap on $(t_1, t_2)$ when both tokens are happy, which would not happen according to \HappySwapAlgo. Then WLOG we can assume $t_1$ was shoved to $v_1$ before $t_2$ was shoved to $v_2$. But this is a contradiction, since $t_2$ can only be shoved away from its destination vertex $v_1 = v_f(t_2)$, which is occupied by $t_1$ at the time of the shove on $t_2$. \\ 

    \textbf{Case II:} The swap on $(t_1, t_2)$ is a shove.\\
    WLOG we may assume that $t_1$ is a happy token, i.e., $v_f(t_1) = v_1$, hence the only move that can be inevitable is $t_2 \rightarrow v_1$, and we will prove that this is indeed the case. Suppose, for the sake of contradiction, that $t_2 \rightarrow v_1$ is redundant. By the definition of a shove we know that $v_2 \neq v_f(t_2)$, then from this, together with Corollary \ref{distance_one}, we know that $t_2$ must reach $v_f(t_2)$ after this swap, but this is a contradiction since $v_1 = v_f(t_1)$.
\end{proof}

\begin{lemma} [Lower Bound on OPT]
\label{lower_bound}
        For \WeightedTokenSwapping\ , $\textrm{OPT} \geq \sum\limits_{t\in T}(\omega(t)\cdot d(t))$.
        
\end{lemma}
\begin{proof}
Note that every algorithm that solves \WeightedTokenSwapping\ needs to move each token $t$ from $v_0(t)$ to $v_f(t)$, which costs at least $\omega(t)\cdot d(v_0(t), v_f(t)) = \omega(t)\cdot d(t)$, hence the lemma follows by summing over each token $t \in T$.
\end{proof}

Now we are ready to prove Theorem \ref{treealg}.
    Denote ALG as the cost induced by \HappySwapAlgo \ on the \WeightedTokenSwapping\ instance and OPT as the cost induced by the optimal solution for the \WeightedTokenSwapping 
  instance. We will prove that $\dfrac{\textrm{ALG}}{\textrm{OPT}} \leq 1+\dfrac{W}{w}$.
 
    From Lemma $\ref{lower_bound}$ we know that $\textrm{OPT} \geq \sum\limits_{t\in T}(\omega(t)\cdot d(t))$. Now we will exhibit an upper bound on ALG. We know that $\textrm{ALG} = \sum\limits_{\text{swap($t$, $t'$) }}\omega(t)+\omega(t')$, where the swaps are induced by \HappySwapAlgo. By Lemma \ref{upper_bound}, we can index each individual summand according to each token's inevitable swaps, i.e., if we set $\tau_t^i$ to be the incident token on the $i$-th inevitable swap of token $t$, then $\textrm{ALG} = \sum\limits_{t\in T}\sum\limits_{i=1}^{d(t)}(\omega(t)+\omega(\tau_t^i)) \leq \sum\limits_{t\in T}\sum\limits_{i=1}^{d(t)}(\omega(t)+W) 
\leq \textrm{OPT}+\sum\limits_{t\in T}(W\cdot d(t))$. Now we have $$\dfrac{\textrm{ALG}}{\textrm{OPT}} \leq 1+ \dfrac{W \cdot \sum\limits_{t\in T}d(t)}{\sum\limits_{t\in T}(\omega(t)\cdot d(t))} \leq 1 + \dfrac{W}{w}$$ where we used the fact that  $\omega(t) \geq w$.\\

\subsection{$O(n^{1-\varepsilon})$-straying algorithms do not achieve better than a ($1 + \frac{W}{w}$)-approx.}

In this section we will prove the following theorem:

\treelb*

Recall that an algorithm for \TokenSwapping\ is \textbf{$\ell$-straying} if the following property holds: At any time during the execution of the algorithm, $\min_{v \in p(t)} d(v(t), v) \leq \ell$ for any $t \in T$. In other words, the distance of any token $t$ from $p(t)$ is at most $\ell$.

% On the other hand, we will show that the above bound is tight for a large class of algorithms. To be precise, we will study the best approximation ratio of $l$-straying algorithms on trees. 

%Note that \HappySwapAlgo\ is a 1-straying algorithm.
% \begin{lemma}
%     For any constants $\varepsilon, \delta > 0$, there does not exist an $O(n^{1-\varepsilon})$-straying algorithm of \WeightedTokenSwapping\  on trees that achieves an approximation ratio better than $1 + \frac{W}{w}-\delta$. 
% \end{lemma}

\begin{proof}
    We will construct an instance of \WeightedTokenSwapping\ as follows, see Figure \ref{fig:straying-approximation}. Consider a path $P$ with two stars attached on both ends. Suppose $\ell = O(n^{1-\varepsilon/2})$, and $P$ has $\ell$  vertices in total, $h_1, \cdots, h_{\ell}$, where each of them is occupied by a token with weight $W$ that is happy originally. Each of the stars contains $N=\frac{n-\ell}{2}$ leaf vertices in total, namely, $v_{1,1},\cdots, v_{1,N}$ and $v_{2,1},\cdots, v_{2,N}$, and each of these vertices is initially occupied by a token with weight $w$. The goal is to exchange the tokens initially on $v_{1,i}$ with $v_{2,i}$ for $i=1,\cdots, N$ (while the tokens on the $h_i$'s begin on their destinations). Since all $v_{2,j}$'s are symmetric, this \WeightedTokenSwapping\ instance is the same up to permuting the index $i$ in $v_{1,i}$'s and permuting the index $j$ in $v_{2,j}$'s.\\

    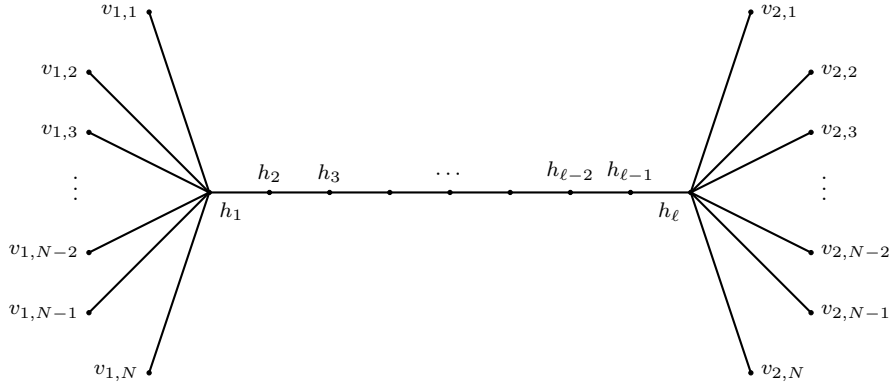
\begin{figure}[h!]
\begin{center}
\begin{tikzpicture}[scale=0.8, line cap=round, line join=round, >=triangle 45]

% \begin{axis}[
%     axis lines=middle,
%     xmajorgrids=true,
%     ymajorgrids=true,
%     xmin=-10, xmax=8,
%     ymin=-11, ymax=9,
%     xtick={-10,-8,...,6},
%     ytick={-10,-8,...,8},
%     enlargelimits=false,
%     axis on top=true,
%     clip=true,
% ]

% h_1 and h_{2l}
\filldraw (-4,0) circle (1pt);
\node[anchor=north west, font=\scriptsize] at (-4,0) {$h_{1}$};

\filldraw (4,0) circle (1pt);
\node[anchor=north east, font=\scriptsize] at (4,0) {$h_{\ell}$};

% Other labeled h_i nodes
\foreach \x/\label in {-3/h_2, -2/h_3, 2/h_{\ell-2}, 3/h_{\ell-1}} {
    \filldraw (\x,0) circle (1pt);
    \node[anchor=south, font=\scriptsize] at (\x,0) {$\label$};
}

% Midpoints without label
\foreach \x in {-1,0,1} {
    \filldraw (\x,0) circle (1pt);
}

% "..." between h_3 and h_{2l-2}
\node at (0,0.3) {\scriptsize $\cdots$};

% Edges
\draw [line width=0.8pt] (-5,3) -- (-4,0);
\draw [line width=0.8pt] (-4,0) -- (4,0);
\draw [line width=0.8pt] (-6,2) -- (-4,0);
\draw [line width=0.8pt] (-6,1) -- (-4,0);
\draw [line width=0.8pt] (-6,-1) -- (-4,0);
\draw [line width=0.8pt] (-6,-2) -- (-4,0);
\draw [line width=0.8pt] (-5,-3) -- (-4,0);

\draw [line width=0.8pt] (4,0) -- (5,3);
\draw [line width=0.8pt] (6,2) -- (4,0);
\draw [line width=0.8pt] (6,1) -- (4,0);
\draw [line width=0.8pt] (6,-1) -- (4,0);
\draw [line width=0.8pt] (6,-2) -- (4,0);
\draw [line width=0.8pt] (4,0) -- (5,-3);

% Left subtree
\filldraw (-5,3) circle (1pt);
\node[anchor=east, font=\scriptsize] at (-5,3) {$v_{1,1}$};

\filldraw (-6,2) circle (1pt);
\node[anchor=east, font=\scriptsize] at (-6,2) {$v_{1,2}$};

\filldraw (-6,1) circle (1pt);
\node[anchor=east, font=\scriptsize] at (-6,1) {$v_{1,3}$};

\node[anchor=east, font=\scriptsize] at (-6,0.2) {$\vdots$}; % between v_{1,3} and v_{1,N-2}

\filldraw (-6,-1) circle (1pt);
\node[anchor=east, font=\scriptsize] at (-6,-1) {$v_{1,N-2}$};

\filldraw (-6,-2) circle (1pt);
\node[anchor=east, font=\scriptsize] at (-6,-2) {$v_{1,N-1}$};

\filldraw (-5,-3) circle (1pt);
\node[anchor=east, font=\scriptsize] at (-5,-3) {$v_{1,N}$};

% Right subtree
\filldraw (5,3) circle (1pt);
\node[anchor=west, font=\scriptsize] at (5,3) {$v_{2,1}$};

\filldraw (6,2) circle (1pt);
\node[anchor=west, font=\scriptsize] at (6,2) {$v_{2,2}$};

\filldraw (6,1) circle (1pt);
\node[anchor=west, font=\scriptsize] at (6,1) {$v_{2,3}$};

\node[anchor=west, font=\scriptsize] at (6,0.2) {$\vdots$}; % between v_{2,3} and v_{2,N-2}

\filldraw (6,-1) circle (1pt);
\node[anchor=west, font=\scriptsize] at (6,-1) {$v_{2,N-2}$};

\filldraw (6,-2) circle (1pt);
\node[anchor=west, font=\scriptsize] at (6,-2) {$v_{2,N-1}$};

\filldraw (5,-3) circle (1pt);
\node[anchor=west, font=\scriptsize] at (5,-3) {$v_{2,N}$};

% \end{axis}
\end{tikzpicture}
\end{center}
\caption{$O(n^{1-\varepsilon})$-straying algorithms do not perform better than $(1+\frac{W}{w}-\delta)$-approximation.}
\label{fig:straying-approximation}
\end{figure}

\paragraph{Upper Bound on OPT}
    We may exhibit an upper bound of OPT by proposing the following solution to \WeightedTokenSwapping\ with three stages: 
    \subparagraph{Stage I} Move the tokens on $h_1,\cdots, h_{\ell}$ to $\ell$ leaves $v_{1,1},\cdots, v_{1,\ell}$ in that order. Each token on $h_i$ will travel a distance of $i$ and each swap is of cost $(W+w)$, since in this stage, each swap on the token that is initially on $h_i$ will involve another token that is initially on $v_{1,j}$ for some $j=1,\cdots, N$. Hence the total cost is $\sum\limits_{i=1}^{\ell} (iW+iw) = \frac{1}{2}\ell(\ell+1)(W+w)$. Note that after Stage I, $h_i$ is occupied by a token that was initially on $v_{1,\ell + 1 - i}$.
    \subparagraph{Stage II} In this stage, we will bring all the tokens with weights $w$ to their destinations by a series of happy swaps, except for those tokens with destinations being one of $v_{1,1},\cdots, v_{1,\ell}$ (which are occupied by tokens with weight $W$ now).
    More precisely, each token on $h_i$ will travel a distance of $\ell + 1 - i$ to reach its destination, which contributes $\sum\limits_{i=1}^{\ell} (\ell + 1 - i)$ to the total traveled distance in this stage, and each token on $v_{1,j}$ for $j = \ell + 1, \cdots, N$ will travel a distance of $\ell  +1$ to reach its destination $v_{2,j}$, which contributes $(N-\ell)(\ell+1)$ to the total traveled distance in this stage. Additionally, it is easy to see that each swap in this stage is a happy swap comprising two tokens, one was initially on $v_{1,i}$'s at the beginning of Stage I, one was initially on $v_{2,i}$'s at the beginning of Stage I, and the swap brings both of them closer to their destinations. Observe that the above discussion on the contribution to the total distance is from the perspective of tokens which were initially on $v_{1,i}$'s at the beginning of Stage I, and no two such tokens swapped with each other in this stage. So in order to get the total distance that all tokens travel in stage II, we may multiply the above distance by 2, to account for the contribution from the tokens that were initially on $v_{2,i}$'s at the beginning of Stage I. Hence the total distance that all tokens travel in stage II is $2(N-\ell)(\ell+1) + 2\sum\limits_{i=1}^{\ell} (\ell + 1 - i) = 2(N-\ell)(\ell+1) + \ell(\ell+1)$. Since each swap in this stage is a happy swap, each swap decreases the total distance by exactly 2, and costs $(w+w) = 2w$, which results in a cost of $(2(N-\ell)(\ell+1) + \ell(\ell+1))w$. 
    \subparagraph{Stage III} Move the tokens with weights $W$ to their original positions by a series of swaps with cost $(W+w)$, which costs another $\frac{1}{2}\ell(\ell+1)(W+w)$. This is completely analogous to Stage I, which can be seen as a reverse process, hence they have the same cost. Note that all tokens have reached their destinations since the tokens on the $h_i$'s return to their original locations, while the tokens on the leaves have exchanged their positions.%, which can be observed by simply tracking the movements of each token. 

    Adding the cost of the 3 stages together we have $\textrm{OPT} \leq \ell(\ell+1)(W+w) + (2(N-\ell)(\ell+1) + \ell(\ell+1))w$.

\paragraph{Lower Bound on ALG}
    Next we obtain a lower bound of $\textrm{ALG}$ for every $k$-straying algorithm, where $k = O(n^{1-\varepsilon})$. According to the definition of $k$-straying algorithm, only those $2k$ vertices that are within distance $k$ from the ends of the path $P$ have the opportunity to be moved into leaves, in other words, there must be $(\ell-2k)$ vertices which are not able to leave path $P$. This implies that, in order to reach the final state, there must be at least $2N(\ell-2k)$ swaps that cost $(W+w)$ each, since every token initially on the leaves must traverse the entire path $P$ to the other side,  hence $\textrm{ALG} \geq 2N(\ell-2k)(W+w)$.

    Now we have 

    \begin{align*}
        \frac{\textrm{ALG}}{\textrm{OPT}} &\geq \frac{2N(\ell-2k)(W+w)}{\ell(\ell+1)(W+w) + (2(N-\ell)(\ell+1) + \ell(\ell+1))w} \\
        &= \frac{2N\ell(W + w) + o(N\ell)}{2N\ell w + o(N\ell)} \\
        &> 1+W/w-\delta \text{ \ \ for any constant $\delta$ and sufficiently large $n$}
    \end{align*}
    where we used the fact that $O(\ell^2) = O(n^{2-\varepsilon})$ and hence is dominated by $O(N\ell) = O(n^{2-\frac{\varepsilon}{2}})$, and $O(\ell) = O(n^{1-\varepsilon/2}) > O(n^{1-\varepsilon}) = O(k)$.
    \end{proof}

\section{Results for General Graphs}

%  In this section, we aim to seek an analogue of the above results on general graphs. More precisely, we shall see a soundness result and a completeness result about a large class of algorithms for \WeightedTokenSwapping\, eventually we will see the result for \TokenSwapping\  from \cite{hiken2024improvedhardnessofapproximationtokenswapping} can be subsumed by our results on \WeightedTokenSwapping\  by setting $W=w$.
\subsection{A polynomial time $(2+2\frac{W}{w})$-approx. algorithm}
In this section we will prove the following theorem:

\genalg*

\subsubsection{Algorithm}

%Now we aim to show that there exists a locally optimal algorithm such that the above bound is tight on arbitrary general graph instances of \WeightedTokenSwapping. 

Our algorithm is a mild extension of the \textbf{Cycle Algorithm} for general graphs (See Appendix A of \cite{hiken2024improvedhardnessofapproximationtokenswapping}) from \TokenSwapping\  to \WeightedTokenSwapping. In particular, our \textbf{Extended Cycle Algorithm} is the same as the Cycle Algorithm for general graphs, except in each cycle we choose the \emph{smallest weight} token to go around the cycle (instead of an arbitrary token in the cycle). %the token that has the smallest weight to pick up the other tokens in the same cycle.
The following is a formal description of the algorithm.

We will denote $\pi: T \xrightarrow[]{} T$ as the permutation induced by $v_0$ and $v_f$ on the set of tokens $T$. Namely, token $t_i\in T$ has final destination vertex that is currently occupied by token $\pi(t_i)$.
\begin{algorithm}
    \caption{Extended Cycle Algorithm}
\begin{algorithmic}[1]
\Function{\ExtendedCycleAlgo}{\WeightedTokenSwapping($G(V,E), T, v_0, v_f, \omega$)}:
\State Decompose $\pi$ into disjoint cycles: $\pi = C_1 \circ C_2\circ \cdots \circ C_r$  where cycle $C_j$ has length $\ell_j$, token set $\{t_{j,1}, \cdots, t_{j,\ell_j}\}$ and $\pi(t_{j,i}) = t_{j,i+1}$, with indices modulo $\ell_j$. Denote $T_{j,k}$ as the set of tokens on a fixed shortest path between $t_{j,k}$ and $t_{j,k+1}$ where $|T_{j,k}| = d(t_{j,k}, t_{j,k+1}) - 1$
\For{$j=1, \cdots, r$}
\State $\ell_j \gets \arg \min_{k=1}^{\ell_j} \{\omega(t_{j,k})\}$
    \For{$k=\ell_j - 1, \cdots, 1$}
    \State move $t_{j,k}$ along its shortest path to its destination \Statex \hspace{6em} (i.e.~$t_{j,k}$ swaps with all tokens in $T_{j,k}$ and then $t_{j,\ell_j}$).
    \State move $t_{j,\ell_j}$ along the same path in the opposite direction 
    \Statex \hspace{6em} (i.e. $t_{j,\ell_j}$ swaps with all tokens in $T_{j,k}$ and reaches $v_0(t_{j,k})$).
   % exchange $t_{j,\ell_j}$, $t_{j,k}$ along the shortest path connecting $v(t_{j,\ell_j})$, $v(t_{j,k})$ by swapping with tokens $T_{j,k}$. 
    \EndFor
\EndFor
\EndFunction
\end{algorithmic}
\end{algorithm}

Since 
%finding the minimum weights for each cycle runs in polynomial time and 
the Cycle Algorithm for general graphs correctly solves \TokenSwapping\ and runs in polynomial time \cite{hiken2024improvedhardnessofapproximationtokenswapping}, it suffices to show that \ExtendedCycleAlgo\ is a $(2 + 2\frac{W}{w})$-approximation for \WeightedTokenSwapping.

\subsubsection{Approximation Analysis}
%\begin{proof}
    Throughout the proof, the indices are to be understood modulo cycle length $l_j$. We can exhibit a lower bound on OPT as follows:
    \begin{align*}
        \textrm{OPT} &\geq \sum_{j=1}^r \sum_{k=1}^{\ell_j} (\omega(t_{j,k})\cdot d(t_{j,k}, t_{j,k+1}))
    \end{align*}
    due to the fact that each token $t_{j,k}$ has weight $\omega(t_{j,k})$, and destination vertex that is occupied by token $t_{j,k+1}$. For the sake of simplicity we denote $s_j := \sum_{k=1}^{\ell_j} (\omega(t_{j,k})\cdot d(t_{j,k}, t_{j,k+1}))$. 
    
    We can express ALG explicitly:
    \begin{align*}
        \textrm{ALG} &= \sum_{k=1}^{\ell_j - 1} \big( (\omega(t_{j,k}) + \omega(t_{j,\ell_j}))\cdot d(t_{j,k}, t_{j,k+1}) + 2\cdot \sum_{k=1}^{\ell_j-1}\sum_{t \in T_{j,k}} \omega(t)\big)\\
    \end{align*}
Indeed, it follows from the fact that $t_{j,k}$ and $t_{j, l_j}$ both travel a distance of $d(t_{j,k}, t_{j,k+1})$ and each of the tokens in $T_{j,k}$ moves twice, once to swap with each of $t_{j,k}$ and $t_{j, l_j}$. Now we exhibit an upper bound of ALG as follows, where we use the fact that  $\omega(t_{j,\ell_j}) = \min_{k=1}^{\ell_j} \{\omega(t_{j,k})\}$ and trivially bound $\omega(t_i) \leq W$:
        \begin{align*}
             & \sum_{k=1}^{\ell_j - 1} \big( (\omega(t_{j,k}) + \omega(t_{j,\ell_j}))\cdot d(t_{j,k}, t_{j,k+1}) + 2\cdot \sum_{k=1}^{\ell_j-1}\sum_{t \in T_{j,k}} \omega(t)\big) \\
            &\leq \sum_{k=1}^{\ell_j - 1} (2\cdot\omega(t_{j,k})\cdot d(t_{j,k}, t_{j,k+1})) + 2\cdot \sum_{k=1}^{\ell_j-1}\sum_{t \in T_{j,k}} W\\
            &< 2s_j  +2W \cdot \sum_{k=1}^{\ell_j-1}d(t_{j,k}, t_{j,k+1})
        \end{align*} 

        Hence we have the following bound on the approximation ratio: 
        \begin{align*}
            \frac{\textrm{ALG}}{\textrm{OPT}} &\leq \frac{\sum_{j=1}^r \big( 2s_j  +2W \cdot \sum_{k=1}^{\ell_j-1}d(t_{j,k}, t_{j,k+1})\big)}{\sum_{j=1}^r s_j} \\
            &= 2 + \frac{\sum_{j=1}^r 2W \cdot \big(\sum_{k=1}^{\ell_j-1}d(t_{j,k}, t_{j,k+1})\big)}{\sum_{j=1}^r s_j} \\
            &= 2 + \frac{2W\cdot \sum_{j=1}^r \sum_{k=1}^{\ell_j-1} d(t_{j,k}, t_{j,k+1})}{\sum_{j=1}^r \sum_{k=1}^{\ell_j} (\omega(t_{j,k})\cdot d(t_{j,k}, t_{j,k+1}))} \\
            &< 2 + \frac{2W \sum_{j=1}^r \sum_{k=1}^{\ell_j}d(t_{j,k}, t_{j,k+1}) }{w\sum_{j=1}^r \sum_{k=1}^{\ell_j}d(t_{j,k}, t_{j,k+1})}
            \\
            &= 2 + 2\frac{W}{w}.
        \end{align*}

        \subsubsection{Generalized Local Optimality}\label{sec:opt}

            Recall that an algorithm for \TokenSwapping\  or \WeightedTokenSwapping\  is said to be \textbf{generalized locally optimal} if every swap takes at least one of the two participating tokens closer to either its destination vertex or its starting vertex. 

            In this section, we briefly argue that \ExtendedCycleAlgo\ is generalized locally optimal. First, the swaps of  $t_{j,k}$ on line 6 are permitted under the definition of generalized locally optimality since each swap moves  $t_{j,k}$ closer to its destination. Next, note that at the beginning of each iteration of the loop on line 3, each token in the graph is either on its starting or destination vertex. This means that the swaps of $t_{j,k}$ on line 6 move each token in $T_{j,k}$ off either its starting or destination vertex. Then, the swaps of $t_{j,\ell_j}$ on line 7 move each token in $T_{j,k}$ back onto (i.e.~closer to) its starting or destination vertex. Thus, the line 7 swaps are permitted under the definition of generalized locally optimality. This completes the argument.

\subsection{Generalized locally optimal algorithms don't do better than $(2+2\frac{W}{w})$-approx.}

In this section we will prove the following theorem:

\genlb*

%\begin{definition}
    % Recall that an algorithm for \TokenSwapping\  or \WeightedTokenSwapping\  is said to be \textbf{generalized locally optimal} if every swap between any pair of tokens ($t_1, t_2$) in the algorithm does not move both $t_1$ and $t_2$ further away from their destinations.
%\end{definition}

% The goal of this section is to construct an instance of \WeightedTokenSwapping\  such that locally optimal algorithms will not admit an approximation ratio better than $2 + 2\frac{W}{w} - \delta$, for any constant $\delta > 0$. 

\begin{proof}
We will follow the construction in Section 5.1 of \cite{hiken2024improvedhardnessofapproximationtokenswapping} and extend it by assigning weights of tokens properly, keeping everything else exactly the same. For the remainder of this section, we will intensively use the notation and terms in Section 5.1 of \cite{hiken2024improvedhardnessofapproximationtokenswapping}, including the parameters $p$ and $q$, the \emph{outer cycle} $C^{out}$, the \emph{inner cycles} $C_j^{in}$, \emph{outer tokens}, and \emph{inner tokens}.

%For a detailed construction of the underlying graph in terms of parameters $p,q$, we refer the readers to \cite{hiken2024improvedhardnessofapproximationtokenswapping}.  

Consider the following weight assignment: We assign weight $w$ to every outer token and weight $W$ to every inner token, see Figure 3 in \cite{hiken2024improvedhardnessofapproximationtokenswapping} and Figure \ref{fig:circle-graph} below.

\begin{figure}[H]
\label{genlowfig}
\centering
\hspace*{-3.40cm}
\definecolor{xdxdff}{rgb}{0.49,0.49,1}
\definecolor{ududff}{rgb}{0.3,0.3,1}
\begin{tikzpicture}
% \begin{axis}[
%     axis lines=middle,
%     grid=both,
%     xmin=-11.2, xmax=11.6,
%     ymin=-4.5, ymax=3.8,
%     xtick={-11,-10,...,11},
%     ytick={-6,-5,...,5},
%     width=\textwidth,
%     height=0.75\textwidth,
%     scale only axis=true,
%     axis equal image
% ]
\clip(-11.2,-3.5) rectangle (11.6,3.8);
\draw [line width=2pt] (0,0) circle (3);

% Quarter-circle arcs
\foreach \x/\y/\a/\b in {
  -3/3/-90/0,
   3/3/180/270,
   3/-3/90/180,
  -3/-3/0/90
}{
  \draw[shift={(\x,\y)},line width=2pt]
    plot[domain=\a:\b,variable=\t] ({3*cos(\t)}, {3*sin(\t)});
}

% m labels
\foreach \x/\y in {
  0.55/3.66, 1.12/3.41, 1.65/3.11, 2.18/2.65, 2.61/2.13,
  2.88/1.59, 3.03/1.05, 3.09/0.43
}{
  \node[anchor=north west] at (\x,\y) {$w$};
}

% M labels
\foreach \x/\y in {
  1.40/1.06, 0.13/2.69, 0.95/1.45, 2.08/0.75, 0.42/2.04
}{
  \node[anchor=north west] at (\x,\y) {$W$};
}

% ududff points
\foreach \x/\y in {
  2.94/0.59, 2.77/1.15, 2.49/1.67, 2.12/2.12, 1.67/2.49,
  1.15/2.77, 0.59/2.94, 0/3, -0.59/2.94, -1.15/2.77,
  -1.67/2.49, -2.12/2.12, -2.49/1.67, -2.77/1.15, -2.94/0.59,
  -3/0, 3/0, -2.94/-0.59, -2.77/-1.15, -2.49/-1.67,
  -2.12/-2.12, -1.67/-2.49, -1.15/-2.77, -0.59/-2.94,
  0/-3, 0.59/-2.94, 1.15/-2.77, 1.67/-2.49, 2.12/-2.12,
  2.49/-1.67, 2.77/-1.15, 2.94/-0.59
}{
  \fill[ududff] (\x,\y) circle (2.5pt);
}

% xdxdff points
\foreach \x/\y in {
  -2.05/0.15, -2.06/-0.15, -1.35/0.49, -0.80/0.96, -0.41/1.49,
  -0.15/2.05, 0.15/2.07, 0.40/1.50, 0.83/0.93, 1.44/0.44,
  2.11/0.13, -1.35/-0.49, -0.79/-0.97, -0.40/-1.51, -0.14/-2.09,
  0.14/-2.09, 0.42/-1.47, 0.89/-0.87, 1.46/-0.43, 2.11/-0.14
}{
  \fill[xdxdff] (\x,\y) circle (2.5pt);
}
% \end{axis}
\end{tikzpicture}
\caption{Generalized locally optimal algorithms do not achieve better than a $(2 + 2\frac{W}{w})$-approx.}
\label{fig:circle-graph}
\end{figure}
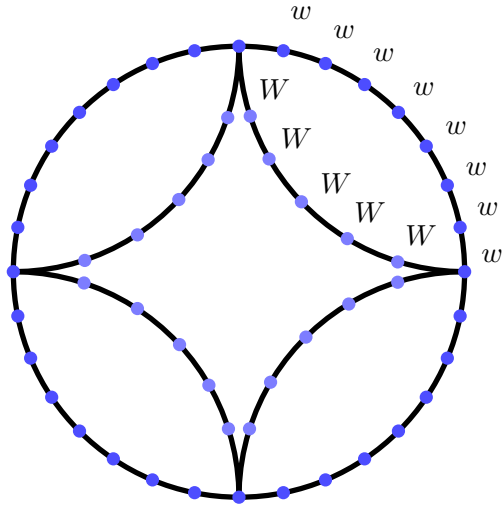

\begin{lemma}
\label{OPTupperbound}
    OPT $\leq 2w\cdot pq^2$
\end{lemma}
\begin{proof}
    We may obtain an upper bound on OPT by showing a solution to our instance of \WeightedTokenSwapping.  Consider the same solution as Claim 5.2 in \cite{hiken2024improvedhardnessofapproximationtokenswapping}, where there are $pq^2$ swaps along $C^{out}$ that solve the \TokenSwapping\  instance, and in all of these swaps both tokens have weight $w$ in our weighted instance.  Since each swap costs $2w$, the total cost of the solution is  $\leq 2w\cdot pq^2$.% henceforth OPT $\leq 2w\cdot pq^2$.
\end{proof}

Now, our goal is to obtain a lower bound on ALG (the cost of a generalized locally optimal algorithm). The following structural result, in the spirit of Claim 5.4 in \cite{hiken2024improvedhardnessofapproximationtokenswapping}, will be useful.
%We will use the following structural result from~\cite{hiken2024improvedhardnessofapproximationtokenswapping}. Note that it remains true for our weighted instance since the definition of a generalized locally optimal algorithm does not depend on weights. %addition of weights does not affect the behavior of a locally optimal algorithm.
\begin{lemma}\label{tokennoexchange}
    For any two tokens $a, b$ that both start on the same inner cycle $C_j^{in}$, generalized locally optimal algorithms only perform swaps of $a,b$ along edges in $C_j^{in}$.
\end{lemma}

\begin{proof}
   % We will show that for any two tokens $t_1, t_2 \in T$, the only possible swaps between $(t_1,t_2)$ are along edges in $C_j^{in}$ for some $j$, hence the lemma follows by noticing that the token set in $C_j^{in}$ remains invariant, for any $j$. 

   It suffices to show that no swaps occur along edges of $C^{out}$. For the sake of contradiction, suppose otherwise and consider the first swap $(t_1,t_2)$ along some edge of $C^{out}$. Note that tokens $t_1$ and $t_2$ must be from adjacent inner cycles, i.e.~$t_1 \in C_j^{in}, t_2 \in C_{j+1}^{in}$ for some $j$. After the swap $(t_1,t_2)$, $t_1$ is on $C_{j+1}^{in}$. Thus, for any $v \in C_{j}^{in}$, the swap $(t_1,t_2)$ increases $d(t_1,v)$ by 1. Indeed, before the swap, the shortest path between $t_1$ and $v$ includes only edges in $C_j^{in}$, and after the swap, the shortest path between $t_1$ and $v$ takes the edge that connects $C_{j}^{in}$ with $C_{j+1}^{in}$ followed by edges in $C_j^{in}$. Since $v_0(t_1), v_f(t_1) \in C_j^{in}$ this implies that $d(t_1, v_0(t_1))$ and $d(t_1, v_f(t_1))$ both increase. Since and $v_0(t_2), v_f(t_2) \in C_{j+1}^{in}$, a symmetric argument shows that $d(t_2, v_0(t_2))$ and $d(t_2, v_f(t_2))$ both increase as well. This violates the definition of generalized local optimality, a contradiction. %in other words, the swap is bringing $t_1$ and $t_2$ away from both of their starting vertices and destination vertices, contradiction!
\end{proof}

Lemma \ref{tokennoexchange} allows us to lower bound the cost of  generalized locally optimal algorithms, by studying the number of swaps within the inner cycles. Now we are ready to show a lower bound on ALG that can be seen as a weighted analogue of Claim 5.5 in \cite{hiken2024improvedhardnessofapproximationtokenswapping}. 

\begin{lemma}
\label{singleinnercyclelowerbound} 
    For any $j$, the cost needed to bring all of the tokens on $C_j^{in}$ to their
target vertices while swapping only along edges in $C_j^{in}$ is at least $(W+w)(2pq-5p/2-4q-\frac{p(p-2)}{8})$. 
\end{lemma}
% \begin{remark}
% Note that our bound in \cref{singleinnercyclelowerbound} asymptotically matches the bound from Claim 5.5 of \cite{hiken2024improvedhardnessofapproximationtokenswapping}, with the appropriate weight multipliers. In particular, setting $W=w=\frac{1}{2}$ (so that each swap in \WeightedTokenSwapping\ contributes 1 to \TokenSwapping), the bounds only differ by $\frac{p(p-2)}{8}$, which is a lower-order term compared to $pq$ if we set $q = \Omega(p)$. This minor difference is because
%     the bound of \cite{hiken2024improvedhardnessofapproximationtokenswapping} includes unweighted swaps which map to weighted swaps of type $(W,W)$, type $(W,w)$ and type $(w,w)$, while our bound only includes swaps of type $(W,W)$ and type $(W,w)$. %which suffices for the final bound. %which contributes to the major term for ALG in \WeightedTokenSwapping, we need to exclude the number of swaps $(W,w)$, which can be upper bounded by $\frac{p}{2}$ and does not affect the major term asymptotically. 
% \end{remark}
\begin{proof}%[Proof of \cref{singleinnercyclelowerbound}]
We begin by recalling Claim 5.5 of \cite{hiken2024improvedhardnessofapproximationtokenswapping}, which asserts that for any $j$, the number of swaps needed to bring all of the tokens on $C_j^{in}$ to their
target vertices while swapping only along edges in $C_j^{in}$ is greater than $2pq-5p/2-4q$. Note that in our weighted setting, there are three types of swaps in our \WeightedTokenSwapping\  instance: $(W,W), (W,w), (w,w)$. Our proof will upper bound the number of swaps of type $(w,w)$ by $\frac{p(p-2)}{8}$, to get a lower bound of the number of swaps of the other two types. Since these two types of swaps each cost at least $(W+w)$, this yields the desired bound.

Recall that $C_j^{in}$ contains $\frac{p}{2}$ outer tokens, all of weight $w$. We may assume for each pair of outer tokens, they swap at most once throughout the solution swap sequence. Indeed, if such a pair swapped twice, we could remove the two swaps from the swap sequence to obtain another valid swap sequence with strictly smaller cost, since the two tokens have the same weight $w$. This implies that the number of swaps over edges of $C_j^{in}$ of type $(w,w)$ is at most ${p/2 \choose 2} = \frac{p(p-2)}{8}$.

Combining this with Claim 5.5 of \cite{hiken2024improvedhardnessofapproximationtokenswapping}, we get the number of swaps of type $(W,w)$ and $(W,W)$ is at least $(2pq-5p/2-4q-\frac{p(p-2)}{8})$. Since each such swap costs at least $(W+w)$, we get the total cost needed to bring all of the tokens on $C_j^{in}$ to their
target vertices while swapping only along edges in $C_j^{in}$ is at least $(W+w)(2pq-5p/2-4q-\frac{p(p-2)}{8})$.
\end{proof}

Now we are ready to complete the proof of Theorem \ref{genlb}.
    There are $2q$ inner cycles in total, and each of them costs at least $(W+w)(2pq-5p/2-4q-\frac{p(p-2)}{8})$ from Lemma $\ref{singleinnercyclelowerbound}$. These costs are additive since by Lemma $\ref{tokennoexchange}$, generalized locally optimal algorithms do not perform any swaps between tokens from different inner cycles. So we get ALG $\geq 2q\cdot (W+w)(2pq-5p/2-4q-\frac{p(p-2)}{8})$. Combining this with Lemma $\ref{OPTupperbound}$, we get 
    \begin{align*}
        \frac{\textrm{ALG}}{\textrm{OPT}} &\geq \frac{2q\cdot (W+w)(2pq-5p/2-4q-\frac{p(p-2)}{8})}{2w\cdot pq^2}\\
        &> 2+2W/w-\delta \text{  \ \ for any constant $\delta$ and sufficiently large $p,q$ and $p= o(q)$.}
    \end{align*}
    Note that it is valid to take sufficiently large $p,q$ and $p= o(q)$, because the only constraints are that $p$ is even, $n = pq + 2pq(q-1) = 2pq^2-pq$, and $p,q$ are at least a certain constant. This concludes the proof of Theorem \ref{genlb}.
    \end{proof}

% \section{Conclusion}
% \input{sections/conclusion.tex}

%\newpage
\bibliographystyle{plain}
\bibliography{references}

\end{document}